\documentclass[9pt,shortpaper,twoside,web]{ieeecolor}
\usepackage{generic}

\usepackage[english]{babel}
\usepackage{amsmath,mathtools,amssymb}
\usepackage{cleveref,cite}
\let\labelindent\relax
\usepackage{enumitem}

\newcommand{\posp}{\text{Span}_+}
\newcommand{\oms}{\Omega_{s+}}

\usepackage{graphicx}
\usepackage{textcomp}

\crefname{app}{Appendix}{Appendices}
\crefname{cor}{Corollary}{Corollary}
\crefname{prop}{Proposition}{Proposition}
\crefname{lemma}{Lemma}{Lemma}
\crefname{defn}{Definition}{Definition}
\crefname{conj}{Conjecture}{Conjecture}
\crefname{exam}{Example}{Example}
\crefname{supp}{Supplemental Section}{Supplemental Section}
\crefname{subsection}{subsection}{subsections}

\newtheorem{theorem}{Theorem}
\newtheorem{cor}{Corollary}
\newtheorem{lemma}{Lemma}

\newcommand{\bs}{\boldsymbol}
\newcommand{\bb}{\mathbb}
\newcommand{\mcal}{\mathcal}

\newcommand{\eye}{\bs{I}}
\newcommand{\zero}{\bs{0}}

\newcommand{\lb}{\left(}
\newcommand{\rb}{\right)}

\newcommand{\lc}{\left\{}
\newcommand{\rc}{\right\}}

\newcommand{\lv}{\left\vert}
\newcommand{\rv}{\right\vert}
\newcommand{\lV}{\left\Vert}
\newcommand{\rV}{\right\Vert}

\newcommand{\LRV}[1]{{\left\vert\kern-0.25ex\left\vert\kern-0.25ex\left\vert #1 \right\vert\kern-0.25ex\right\vert\kern-0.25ex\right\vert}}

\newcommand{\nth}{^\mathsf{th}}
\newcommand{\tran}{^{\mathsf{T}}}

\newcommand{\rank}[1]{\mathsf{Rank}\lc#1\rc}

\newcommand{\matA}{\bs{A}}
\newcommand{\matB}{\bs{B}}

\newcommand{\matJ}{\bs{J}}

\newcommand{\matN}{\bs{N}}

\newcommand{\matP}{\bs{P}}
\newcommand{\matQ}{\bs{Q}}

\newcommand{\matZ}{\bs{Z}}

\newcommand{\bbC}{\bb{C}}

\newcommand{\bbR}{\bb{R}}

\newcommand{\calA}{\mcal{A}}

\newcommand{\calC}{\mcal{C}}

\newcommand{\calN}{\mcal{N}}
\newcommand{\calO}{\mcal{O}}

\newcommand{\calS}{\mcal{S}}

\newcommand{\calV}{\mcal{V}}

\newcommand{\calZ}{\mcal{Z}}

\newcommand{\veca}{\bs{a}}
\newcommand{\vecb}{\bs{b}}

\newcommand{\vecu}{\bs{u}}

\newcommand{\vecx}{\bs{x}}
\newcommand{\vecy}{\bs{y}}
\newcommand{\vecz}{\bs{z}}

\newcommand{\vecalpha}{\bs{\alpha}}

\newcommand{\vecrho}{\bs{\rho}}

\newcommand{\matPhi}{\bs{\Phi}}

\def\BibTeX{{\rm B\kern-.05em{\sc i\kern-.025em b}\kern-.08em
    T\kern-.1667em\lower.7ex\hbox{E}\kern-.125emX}}
\begin{document}
\title{Controllability of a Linear System with Nonnegative Sparse Controls}
\author{Geethu Joseph
\thanks{The author is with the Dept.\ of EECS at the Syracuse University, NY, USA, Email: gjoseph@syr.edu.}}
\maketitle
\begin{abstract}
This paper studies controllability of a discrete-time linear dynamical system using nonnegative and sparse inputs. These constraints on the control input arise naturally in many real-life systems where the external influence on the system is unidirectional, and activating each input node adds to the cost of control. We derive the necessary and sufficient conditions for controllability of the system, without imposing any constraints on the system matrices. Unlike the well-known Kalman rank based controllability criteria, the conditions presented in this paper can be verified in polynomial time, and the verification complexity is independent of the sparsity level.  The proof of the result is based on the analytical tools concerning the properties of a convex cone. Our results also provide a closed-form expression for the minimum number of control nodes to be activated at every time instant to ensure controllability of the system using positive controls.
\end{abstract}
\section{Introduction}
Controllability is one of the most fundamental concepts in control theory, which is related to the ability of a system to maneuver its states. Originally, controllability of a linear system was studied without any constraints on the inputs. These studies led to the complete characterization of controllability through the classical Kalman rank test and Popov-Belovich-Hautus (PBH) test~\cite{kalman1959general,hautus1970stabilization}. Traditional controllability has also been extended to the case where the admissible input set is constrained due to practical limitations. The different types of constraint sets that have been considered in the past are compact, convex, or quasi-convex sets~\cite{brammer1972controllability,schmitendorf1980null,son1981controllability,sontag1984algebraic}. In this paper, we deal with nonconvex and noncompact constraints on the input, namely, sparsity and nonnegativity. More precisely, we consider a linear dynamical system whose control input at every time instant has a few number of nonzero entries compared to its dimension, and all the nonzero entries are positive. For such a system, we investigate whether it is possible to steer the system from an arbitrary initial state to an arbitrary final state within a finite time duration.

\subsection{Motivation}
The sparsity constraint is desired in many real-world applications due to communication bandwidth, cost, or energy constraints~ \cite{joseph2019controllability}. Also, the nonnegativity constraint is frequently encountered in medical, ecological, chemical,
and economical applications where the controls have a unidirectional influence~\cite{saperstone1971controllability,benvenuti2003positive}. To motivate our setting with both sparsity and nonnegativity constraints, consider the problem of controlling the temperature of a large room with multiple heating elements~\cite{heemels1998linear}. The heaters can radiate heat energy into the room, but they cannot extract heat. So, the control is in one-way stream, and thus, the inputs are nonnegative. Also, it is desirable to maintain the temperature by operating as few heaters as possible to reduce the overall cost of operation. Consequently, the control input at every time instant is sparse. %Moreover, repeatedly using the same set of heaters may lead to wearing out of those heaters. Thus, we use different sets of heaters at every time instant to extend the life of the heating system and to provide better control of the temperature. 

\subsection{Related work}
The two constraints that we consider here, sparsity and nonnegativity, have been studied separately in literature for several decades. We provide a short review of these works below:
\subsubsection{Sparsity constraint}
The study of a linear dynamical system under sparsity constraints dates back to 1972~\cite{athans1972determination}. Some recent works have addressed the problem of finding the sequence of sparse control inputs, both for the fixed and the time-varying set of control nodes, and the other related problems~\cite{olshevsky2014minimal,siami2018deterministic,zhao2016scheduling}. However, controllability under sparse inputs was not well-understood, until recently. The widely known condition
for sparse-controllability is the extended version of the Kalman rank test. This test is based on the rank of the Gramian matrix, and it is known to have combinatorial complexity. Hence, different quantitative measures of controllability based on the Gramian matrix have been considered: the smallest eigenvalue, the trace of the inverse, the inverse of the trace, the determinant, the maximum entry in the diagonal, etc.~\cite{pasqualetti2014controllability,chanekar2017optimal, nozari2017time,jadbabaie2018deterministic}. However, these strategies make the analysis cumbersome. Recently, a set of algebraically verifiable necessary and sufficient conditions for sparse controllability that are similar to the classical PBH test~\cite{hautus1970stabilization} were presented in \cite{joseph2019controllability}. Our work analyzes a more constrained system where the inputs are not only sparse but also nonnegative.

\subsubsection{Nonnegativity constraint}
Controllability of linear systems with nonnegative control inputs was first studied in \cite{brammer1972controllability} for continuous-time linear systems. These results were extended to discrete-time systems in \cite{evans1977controllability} for a single input system, and were further extended to multi-input systems in \cite{son1981controllability}. These papers have triggered a large number of studies dealing with the other related problems on controllability like approximate controllability, null controllability, the geometry of the reachability set, etc.~\cite{son1997approximate,benvenuti2006geometry,nguyen1986null}. Other extensions of the controllability result to nonstationary systems and third-order systems have also been investigated~\cite{phat1992constrained,benvenuti2011reachable}. To the best of our knowledge, none of the existing works in the literature has addressed the important problem of controllability of a linear system under both sparsity and nonnegativity constraints. 

\subsection{Our contributions}
We derive a set of necessary and sufficient conditions for controllability of a linear system under
nonnegative sparse inputs. We show that any system is controllable with nonnegative sparse control inputs if and only if it is controllable using nonnegative control inputs and the sparsity level is greater than the dimension of the null space of the state-transition matrix. Using this result, we show that the conditions for verifying controllability are non-combinatorial. Our approach is based on fundamental tools from analysis concerning the properties of positive spanning sets. 

\begin{table}
\caption{Notation}
\label{tab:notation}
\normalsize

\begin{center}
\begin{tabular}{|r|l|}
\hline
$\bbR$ & Set of real numbers\\
$\bbC$ & Set of complex numbers\\
 $\eye$ & Identity matrix\\
 $\zero$ &  All zero matrix (or vector)\\
$\matA_{\calA}$ & Submatrix of $\matA$ formed by the \\
&\hfill columns indexed by the set $\calA$\\
$\Vert  \cdot \Vert_0$ & $\ell_0$ norm of a vector\\
$\lv\cdot\rv$& Cardinality of a set\\
$\geq$ &  Element-wise inequality\\
&\hfill $\veca\geq\vecb$ implies $\veca_i\geq\vecb_i$, $\forall i$\\
$\posp\{\matA\}$ & Positive span of the columns of $\matA$ \\
&$\posp\{\matA\}=\lc\veca:\veca=\sum_{i=1}^m\alpha_i\matA_i,\alpha_i\geq0\rc$\\&\\
\hline
\end{tabular}
\end{center}

\end{table}

\section{Nonnegative Sparse Controllability}
We consider the discrete-time linear dynamical system $(\matA,\matB)$ whose state evolution model is as follows:
\begin{equation}
\vecx_k = \matA\vecx_{k-1}+\matB\vecu_{k},\label{eq:sys}
\end{equation}
where $\vecx_k\in\bbR^{N}$ and $\vecu_k\in\Omega\subseteq\bbR^{m}$ denote the state vector and the control vector at time $k$, respectively. The set $\Omega$ is the set of all admissible controls of the system. Also, $\matA\in\bbR^{N\times N}$ is the state-transition matrix, and $\matB\in\bbR^{N\times m}$ is the input matrix of the system. We assume that the control vectors are $s$-sparse and its nonzero entries are positive:
\begin{equation}\label{eq:constraint}
\Omega=\oms\triangleq \lc\vecz\in\bbR^m:\lV\vecz\rV_0\leq s \text{ and } \vecz \geq \zero\rc\subset\bbR^m.
\end{equation}
Our goal  is to examine controllability of the system, i.e., for any given pair $(\vecx_{\mathrm{initial}}\in\bbR^N,\vecx_{\mathrm{final}}\in\bbR^N)$, we test if  it is possible to find inputs from $\Omega$ such that $\vecx_K=\vecx_{\mathrm{final}}$ when $\vecx_0=\vecx_{\mathrm{intial}}$, for some  finite positive integer $K$. This notion of controllability is  henceforth referred to as \emph{nonnegative sparse controllability}. 

From \eqref{eq:sys}, the state vector at time $K$ is given by
\begin{equation}\label{eq:sys_concat}
\vecx_{K}-\matA^K\vecx_0 = \sum_{k=1}^{K}\matA^{K-k}\matB\vecu_k.
\end{equation}
Therefore, the system is controllable if and only if there exists a positive integer $K<\infty$ such that 
\begin{equation}\label{eq:combinatorial}
\bigcup_{\{\calS_k:\lv\calS_k\rv\leq s\}_{k=1}^K} \hspace{-0.5cm}\posp\lc \begin{bmatrix}\matA^{K-1} \matB_{\calS_1} & \matA^{K-2} \matB_{\calS_2} \ldots  \matB_{\calS_K}\end{bmatrix}\rc=\bbR^N,
\end{equation}
where $\posp$ is defined in \Cref{tab:notation}.
A brute force verification of the above condition is combinatorial, and hence, it is computationally heavy. In the sequel, we present a non-combinatorial verification procedure for testing controllability using nonnegative sparse controls. 

\section{Preliminaries}
Our system imposes two types of constraints on the set of admissible inputs (as given in \eqref{eq:constraint}): nonnegativity and sparsity. These two constraints have been separately dealt in the literature and we present the corresponding results below:
\renewcommand{\thetheorem}{\Alph{theorem}}
\renewcommand{\thecor}{\Alph{theorem}}
\begin{theorem}[{~\cite[Theorem 1]{son1981controllability}}]\label{thm:nonnegative}
The system $(\matA,\matB)$ defined in \eqref{eq:sys} is controllable using nonnegative control inputs if and only if the following conditions hold:
\begin{enumerate}[label=(\roman*)]
\item $\nexists\lb \lambda\in\bbC,\vecz\neq\zero\rb$ such that $\vecz\tran\matA=\lambda\vecz\tran$ and $\vecz\tran\matB=\zero$.\label{con:control_non}
\item $\nexists\lb\lambda\geq 0,\vecz\neq\zero\rb$ such that $\vecz\tran\matA=\lambda\vecz\tran$ and $\vecz\tran\matB\leq \zero$.\label{con:nonnegative_non}
\end{enumerate}
\end{theorem}
Next, we present the results for controllability using sparse vectors:
\begin{theorem}[{~\cite[Theorem 1]{joseph2019controllability}}]\label{thm:sparse}
The system $(\matA,\matB)$ defined in \eqref{eq:sys} is controllable using $s-$sparse control inputs if and only if the following conditions hold:
\begin{enumerate}[label=(\roman*)]
\item $\nexists\lb \lambda\in\bbC,\vecz\neq\zero\rb$ such that $\vecz\tran\matA=\lambda\vecz\tran$ and $\vecz\tran\matB=\zero$.
\item The sparsity $s\geq N-\rank{\matA}$.
\end{enumerate}
\end{theorem}
\renewcommand{\thetheorem}{\arabic{theorem}}
\setcounter{theorem}{0}
In the next section, we present the main result of the paper and the insights that it yields.
\section{Necessary and Sufficient Conditions}\label{sec:necc_suff}
The section presents the necessary and sufficient conditions for controllability of the system in \eqref{eq:sys} under the constraint $\Omega=\oms$ defined in \eqref{eq:constraint}. The constraint $\Omega=\oms$ is more restrictive than both the nonnegativity constraint in \Cref{thm:nonnegative} and the sparsity constraint in \Cref{thm:sparse}. Thus, the conditions of \Cref{thm:nonnegative} and \Cref{thm:sparse} are necessary for the case when $\Omega=\oms$. The next result shows that those conditions are not only necessary but also sufficient for controllability under the constraint $\Omega=\oms$.
\begin{theorem}\label{thm:necc_suff}
Suppose that the set of admissible vectors $\Omega$ is given by \eqref{eq:constraint} with $s>0$. Then, the system $(\matA,\matB)$ defined in \eqref{eq:sys} is controllable if and only if the following conditions hold:
\begin{enumerate}[label=(\roman*)]
\item \label{con:control}$\nexists \lb \lambda\in\bbC,\vecz\neq\zero\rb$ such that $\vecz\tran\matA=\lambda\vecz\tran$ and $\vecz\tran\matB=\zero$.
\item \label{con:nonnegative} $\nexists\lb\lambda\geq 0,\vecz\neq\zero\rb$ such that $\vecz\tran\matA=\lambda\vecz\tran$ and $\vecz\tran\matB\leq 0$.
\item \label{con:sparse} The sparsity $s\geq N-\rank{\matA}$.
\end{enumerate}
\end{theorem}
\begin{proof}
See \Cref{app:necc_suff}.
\end{proof}
The immediate observations from the above result are as follows:
\begin{itemize}
\item \Cref{thm:necc_suff} implies that any $s-$sparse controllable system is nonnegative $s-$sparse controllable if and only if it satisfies Condition \ref{con:nonnegative} of \Cref{thm:necc_suff}. This is evident from \Cref{thm:sparse}. Similarly, from \Cref{thm:nonnegative}, if a linear system  is controllable using nonnegative control inputs, it is  nonnegative $s-$sparse controllable if and only if $s\geq N-\rank{\matA}$. Therefore, the extra condition for ensuring the sparsity of the control inputs is independent of the input matrix $\matB$.
\item For the special case when $s=m$, \Cref{thm:necc_suff} reduces to \Cref{thm:nonnegative}, as expected. Also, when $m=1$, the notion of sparse controllability and controllability are same, and hence, \Cref{thm:necc_suff} reduces to the well-known result of Evans and Murthy~\cite[Theorem 1]{evans1977controllability}.
\item If the linear system in \eqref{eq:sys} is nonnegative $s-$sparse controllable, it is also nonpositive $s-$sparse controllable. This is because if the system given by $(\matA,\matB)$ satisfies the conditions of \Cref{thm:necc_suff}, the system given by $(\matA,-\matB)$ also satisfies those conditions. In particular, $\nexists\vecz\neq\zero$ such that $\vecz\tran\matA=\lambda\vecz\tran$ and $\vecz\tran\matB\geq 0$, for some $\lambda\geq 0$. This follows since for every $\vecz$ such that $\vecz\tran\matA=\lambda\vecz\tran$, we have $(-\vecz)\tran\matA=\lambda(-\vecz)\tran$.
\item For the linear system in \eqref{eq:sys}, controllability using nonnegative and sparse inputs with a common support (i.e., the positive entries of all control inputs coincide) holds only if $s\geq N-\rank{\matA}$. This follows from Condition \ref{con:sparse} of \Cref{thm:necc_suff} because here, the control signals are more restricted than the setting in \Cref{thm:necc_suff}.
\end{itemize}

We obtain the following interesting corollary from \Cref{thm:necc_suff}.
\begin{cor}\label{cor:sbound}
If any system as given in \eqref{eq:sys} is nonnegative controllable, then it is nonnegative $s$-sparse controllable if $s\geq m-1$.
\end{cor}
\begin{proof}
See \Cref{app:sbound}.
\end{proof}

\subsection{Computational complexity}
In this subsection, we discuss the computational complexity of the controllability test given in \Cref{thm:necc_suff}:
\begin{itemize}
\item To check Condition \ref{con:control} of \Cref{thm:necc_suff}, we solve for all eigenvalues of $\matA$ and check if $\rank{\begin{bmatrix}
\lambda\eye-\matA & \matB
\end{bmatrix}}=N$, for each eigenvalue $\lambda$. So, the complexity of this step is polynomial in $N$ and $m$.
\item To check Condition \ref{con:nonnegative}  of \Cref{thm:necc_suff}, for every eigenvalue $\lambda\geq0$ of $\matA$, we find a set of linearly independent eigenvectors $\lc\vecz_i^{(\lambda)}\rc_{i=1}^{g_{\lambda}}$ corresponding to $\lambda$,  where $g_{\lambda}$ denotes its geometric multiplicity. 
Now,  Condition \ref{con:nonnegative}  of \Cref{thm:necc_suff} can be verified by checking if there exists $\vecrho\in\bbR^{{g_{\lambda}}}$ such that $\vecrho\tran\matZ\tran\matB\leq \zero$. Here, $\matZ\in\bbR^{N\times g_{\lambda}}$ is a matrix formed by the vectors $\lc\vecz_i^{(\lambda)}\rc_{i=1}^{g_{\lambda}}$.  The feasibility of the set of linear inequalities $\matB\tran\matZ\vecrho\leq \zero$ can be verified by solving the following (dummy) linear programming problem:
\begin{equation}
\underset{\rho\in\bbR^{g_{\lambda}}}{\max} 0 \; \text{ subject to }\; \matB\tran\matZ\vecrho\leq \zero.
\end{equation}
Thus, the complexity of verification of  Condition \ref{con:nonnegative}  of \Cref{thm:necc_suff} is also polynomial in $N$ and $m$.
\item The complexity to verify Condition \ref{con:sparse}  of \Cref{thm:necc_suff} is independent of the system dimension: $\calO(1)$.
\end{itemize}
Therefore, the overall complexity of our controllability test is non-combinatorial unlike the verification of condition \eqref{eq:combinatorial}. Moreover, the complexity is independent of the sparsity level $s$ whereas the complexity of the verification of condition \eqref{eq:combinatorial} is a function of $\binom{N}{s}$.
\subsection{Comparison with sparse-controllability}
From \Cref{thm:sparse} and \Cref{thm:necc_suff}, we see some similarities between sparse-controllability and nonnegative sparse controllability. 
\begin{itemize}
\item \emph{Reversible systems: }If the system is reversible, i.e., state-transition matrix $\matA$ is nonsingular, it is $s-$sparse-controllable for any $0 < s \leq m$ if and only if it is controllable ($\Omega=\bbR^m$). Similarly, it is nonnegative $s-$sparse-controllable for any $0 < s \leq m$ if and only if it is controllable using nonnegative controls ($\Omega=\bbR^m_+$).
\item \emph{Minimal control: } Suppose that the system defined by the matrix pair $(\matA,\matB_{\calS})$ is controllable under the nonnegativity constraint $\Omega=\bbR^m_+$, for some index set $\calS\subseteq \{1,2,\ldots,m\}$. Then, the system is nonnegative $s-$sparse-controllable. In particular, if $\rank{\matB}\leq s$, controllability under the constraint $\Omega=\bbR_+^N$ implies nonnegative $s-$sparse controllability. Sparse-controllability also posses a similar property.
\item \emph{Change of basis: }If a system as given in \eqref{eq:sys} is controllable using inputs that are $s-$sparse under the canonical basis, it is controllable using inputs that are $s-$sparse under any  basis $\Phi\in\bbR^{m\times m}$. However, this property does not hold for nonnegative sparse controllability. For example, consider the following system:
\begin{align}
\matA=\begin{bmatrix}
-1 & 0 & 0\\
0 &-1&0\\
0&0&0
\end{bmatrix}, \matB&=\begin{bmatrix}
1 & 0&0 &0\\
0&1&0&0\\
0&0&1&-1
\end{bmatrix} \\
\matPhi=\begin{bmatrix}
1&0&0&0\\
0&1 & 1&0\\
0&-1&-1&0\\
0&0&0&1
\end{bmatrix}.
\end{align}
The system ($\matA,\matB$) is nonnegative $1-$sparse controllable, but the system ($\matA,\matB\matPhi$) is not nonnegative $1-$sparse controllable. This is because the eigenvector $\begin{bmatrix}
0 & 0 & 1
\end{bmatrix}\tran$ corresponding the eigenvalue 0 of $\matA$ does not satisfy Condition \ref{con:nonnegative} of \Cref{thm:necc_suff}.
\end{itemize}
\section{Conclusions}
This paper characterized controllability of discrete-time linear systems subject to sparsity and nonnegativity constraints on the inputs. The characterizations are in terms of algebraic conditions that are similar to the classical results for unconstrained and nonnegative input-constrained linear systems. We showed that the complexity of the controllability test is polynomial in the dimensions of the system. Extending our results to the other notions of controllability like output controllability, approximate controllability, etc., is a related direction of research. Also, this paper dealt with the theoretical question of the existence of a sequence of nonnegative sparse vectors that ensure controllability. Developing computationally efficient algorithms to find the energy-optimal sequence of such nonnegative sparse vectors can be an interesting avenue for future work. 
\appendices
\crefalias{section}{appendix}

\section{Proof of \Cref{thm:necc_suff}}\label{app:necc_suff}
The necessity of the three conditions of \Cref{thm:necc_suff} is straightforward from \Cref{thm:nonnegative} and \Cref{thm:sparse}. Consequently, we need to show that the conditions of \Cref{thm:necc_suff} are sufficient for nonnegative sparse controllability. This proof relies on the following three lemmas.
\begin{lemma}\label{lem:jordan}
For any matrix $\matA\in\bbR^{N\times N}$, there exists integers $0\leq q, n\leq N$ and an invertible matrix $\matP\in\bbR^{N\times N} $ such that
\begin{equation}
\matP = \begin{bmatrix}
\matP_{(0)}\in\bbR^{q \times N}\\
\zero \in\bbR^{N-q\times N}
\end{bmatrix} + \sum_{i=1}^n\matP_{(i)},
\label{eq:jordan1}
\end{equation}
where $\matP_{(0)}$ and $\matP_{(i)}\in\bbR^{N\times N}$ satisfies the following:
\begin{enumerate}[label=(C\arabic*)]
\item The rank of $\matP_{(0)}$ is $q$, and there exists a nonsingular matrix $\matJ\in\bbR^{q\times q}$ such that $\matP_{(0)}\matA^k = \matJ^k\matP_{(0)}$.\label{con:jordan1} 
\item \label{con:jordan2} For all values of $1\leq i\leq n$, 
\begin{align}
\rank{\matP_{(i)}} &= \rank{\matP_{(i)}\matA^{i-1}} \leq N-\rank{\matA}\label{eq:jordan2}\\
\matP_{(i)}\matA^k&=\zero,  k\geq i.\label{eq:jordan3}
\end{align}
\end{enumerate}
\end{lemma}
\begin{proof}
See \Cref{app:jordan}.
\end{proof}
\begin{lemma}\label{lem:nonzero}
Let  $\matA$ be nonsingular. If the system $(\matA,\matB)$ defined in \eqref{eq:sys} is nonnegative controllable, then the system is nonnegative $s$-sparse controllable, for any $s>0$.
\end{lemma}
\begin{proof}
See \Cref{app:nonzero}.
\end{proof}
\begin{lemma}\label{lem:union}
Let $\calZ\subseteq \bbR^N$ be a subspace of dimension $1\leq d\leq N$. Also, let $\matZ\in\bbR^{N\times m}$ where $m > d$ is such that $\calZ=\posp\{\matZ\}$. Then, the following relation holds:
\begin{equation}\label{eq:union}
\calZ=\bigcup_{\calS\subseteq\{1,2,\ldots,m\},\lv\calS\rv=d}\posp\{\matZ_\calS\}.
\end{equation}
\end{lemma}
\begin{proof}
See \Cref{app:union}.
\end{proof}

From \eqref{eq:combinatorial}, it is enough to show that under the conditions of \Cref{thm:necc_suff}, for any $\vecx\in\bbR^N$, there exists a positive integer $K<\infty$ and $\lc\vecu_k\in\oms\rc_{k=1}^K$ such that
\begin{equation}\label{eq:basic_con}
\vecx=\sum_{k=1}^{K}\matA^{K-k}\matB\vecu_k.
\end{equation}
Using \Cref{lem:jordan}, we see that \eqref{eq:basic_con}  is equivalent to the following:
\begin{align}\label{eq:basic_ver11}
\begin{bmatrix}
\matP_{(0)}\\
\zero
\end{bmatrix}\vecx&= \begin{bmatrix}
\matP_{(0)}\\
\zero
\end{bmatrix}\sum_{k=1}^K\matA^{K-k}\matB\vecu_k\\
\matP_{(i)}\vecx  &= \matP_{(i)}\sum_{k=1}^K \matA^{K-k}\matB\vecu_k, \;1\leq i\leq n,\label{eq:basic_ver12}
\end{align}
where $\matP_{(0)}\in\bbR^{q\times N}$ and $\matP_{(i)}\in\bbR^{N\times N}$ are as defined in the statement of \Cref{lem:jordan}. Here, the equivalence follows because the sum of \eqref{eq:basic_ver11} and \eqref{eq:basic_ver12} over all values of $i$ gives \eqref{eq:basic_con}. We remove the zero rows of \eqref{eq:basic_ver11} and simplify \eqref{eq:basic_ver12} to obtain the following equivalent equations:
\begin{align}
\matJ^{-n}\matP_{(0)}\lb \vecx-\sum_{k=K-n+1}^K\matA^{K-k}\matB\vecu_k\rb\notag\\
&\hspace{-1cm}= 
\sum_{k=1}^{K-n}\matJ^{K-n-k}\matP_{(0)}\matB\vecu_{k}\label{eq:col_con}\\
\matP_{(i)}\lb \vecx - \sum_{k=K-(i-2)}^{K}\matA^{K-k}\matB\vecu_k\rb\notag\\
&\hspace{-1cm}=   \matP_{(i)} \matA^{i-1}\matB\vecu_{K-i+1},\label{eq:null_con}
\end{align}
where we use Conditions \ref{con:jordan1} and \ref{con:jordan2} of \Cref{lem:jordan} and $\matJ$ is as defined in the statement of \Cref{lem:jordan}. Further, we notice that
\begin{align}
\matJ^{-n}\matP_{(0)}\lb \vecx-\sum_{k=K-n+1}^K\matA^{K-k}\matB\vecu_k\rb&\in\bbR^q\\
\matP_{(i)}\lb \vecx - \sum_{k=K-(i-2)}^{K}\matA^{K-k}\matB\vecu_k\rb&\in\calC_i,
\end{align}
where  $\calC_i$ is the column space of $\matP_{(i)}$. We also notice that the control inputs $\vecu_k$ appearing on the right-hand sides of \eqref{eq:col_con} and \eqref{eq:null_con} for each value of $1\leq i\leq n$ are different.  Then, to prove nonnegative sparse controllbaility, it suffices to prove the following: 
\begin{enumerate}[label=(\alph*)]
\item For any $\vecz\in\bbR^{q}$, there exists a positive integer $K<\infty$ and $\lc\vecu_k\in\oms\rc_{k=1}^{K}$ which satisfy the relation:
\begin{equation}
\vecz=\sum_{k=1}^K\matJ^{K-k}\matP_{(0)}\matB\vecu_k.
\end{equation}
\label{con:col}
\item For all integers $1\leq i\leq n$ and any $\vecz\in\calC_i$, there exists $\vecu\in\oms$ such that 
\begin{equation}
\vecz=\matP_{(i)}\matA^{i-1}\matB\vecu.
\end{equation}  \label{con:null}
\end{enumerate}
In the remainder of this section, we provide the proof the above two statements.

Proving Statement \ref{con:col} is equivalent to showing that under the conditions of \Cref{thm:necc_suff}, the system $(\matJ,\matP_{(0)}\matB)$ is nonnegative $s-$sparse controllable. We prove the negation of this, i.e., we show that if the system $(\matJ,\matP_{(0)}\matB)$ is not nonnegative $s-$sparse controllable, at least one condition of \Cref{thm:necc_suff} is violated. As $\matJ$ is nonsingular, we apply \Cref{lem:nonzero}  to deduce that if the system $(\matJ,\matP_{(0)}\matB)$ is not nonnegative $s-$sparse controllable, it is not nonnegative controllable. Then, \Cref{thm:nonnegative} implies that at least one of the following hold:
\begin{itemize}
\item $\exists\lb\lambda\in\bbC,\vecy\neq\zero\rb$ with $\vecy\tran\matJ=\lambda\vecy\tran$ and $\vecy\tran\matP_{(0)}\matB=\zero$.
\item $\exists\lb\lambda>0,\vecy\neq\zero\rb$ with $\vecy\tran\matJ=\lambda\vecy\tran$ and $\vecy\tran\matP_{(0)}\matB\leq \zero$.
\end{itemize}
However, if $\vecy\tran\matJ=\lambda\vecy\tran$, then from Condition \ref{con:jordan1} of \Cref{lem:jordan} with $k=1$,
\begin{equation}
\vecy\tran\matP_{(0)}\matA=\vecy\tran\matJ\matP_{(0)} = \lambda\vecy\tran\matP_{(0)}.
\end{equation}
Condition \ref{con:jordan1} of \Cref{lem:jordan} also implies that that $\matP_{(0)}$ is full row rank. Consequently, $\vecz=\matP_{(0)}\tran\vecy\neq\zero$.
Thus, $\vecz=\matP_{(0)}\tran\vecy$ is an  eigenvector  of $\matA$, which leads to atleast one of the following conditions:
\begin{itemize}
\item $\exists\lb\lambda\in\bbC,\vecz\neq\zero\rb$ with $\vecz\tran\matA=\lambda\vecz\tran$ and $\vecz\tran\matB=\zero$.
\item $\exists\lb\lambda>0,\vecz\neq\zero\rb$ with $\vecz\tran\matA=\lambda\vecz\tran$ and $\vecz\tran\matB\leq \zero$.
\end{itemize}
Therefore,  if the system $(\matJ,\matP_{(0)}\matB)$ is not nonnegative $s-$sparse controllable, at least one of the conditions of \Cref{thm:necc_suff} is violated. Hence, Statement \ref{con:col} is proved.

Next, we complete the proof of \Cref{thm:necc_suff} by establishing that Statement \ref{con:null} is true when the conditions of \Cref{thm:necc_suff} hold. Statement \ref{con:null} is equivalent to
\begin{equation}\label{eq:con_b}
\calC_i\subseteq \bigcup_{\calS\subseteq\{1,2,\ldots,m\},\lv\calS\rv=s}\posp\lc\matP_{(i)}\matA^{i-1}\matB_\calS\rc.
\end{equation}
Conditions \ref{con:control} and \ref{con:nonnegative} of \Cref{thm:necc_suff} imply that the system $(\matA,\matB)$ is nonnegative controllable. So, for any $\vecx\in \bbR^N$, there exists a positive integer $K<\infty$ and $\lc\vecu_k\in\bbR^N_+\rc_{k=1}^K$ such that 
\begin{equation}\label{eq:nsparse_con}
\vecx=\sum_{k=1}^{K}\matA^{K-k}\matB\vecu_k.
\end{equation}
Multiplying both the sides of \eqref{eq:nsparse_con} with $\matP_{(i)}\matA^{i-1}$, we get 
\begin{equation}\label{eq:reduced_null}
\matP_{(i)}\matA^{i-1}\vecx =   \matP_{(i)} \matA^{i-1}\matB\vecu_{K},
\end{equation}
where we use \eqref{eq:jordan3} of \Cref{lem:jordan}. Also, from \eqref{eq:jordan2}, the columns spaces of $\matP_{(i)}$ and $\matP_{(i)}\matA^{i-1}$ are the same. Since \eqref{eq:reduced_null} holds for any $\vecx\in\bbR^N$, we deduce that 
\begin{equation}
\calC_i\subseteq \posp\lc\matP_{(i)}\matA^{i-1}\matB\rc.
\end{equation}
Now, using \Cref{lem:union}, we obtain
\begin{equation}
\calC_i\subseteq \bigcup_{\calS\subseteq\{1,2,\ldots,m\},\lv\calS\rv=\dim\{\calC_i\}}\posp\lc\matP_{(i)}\matA^{i-1}\matB_\calS\rc,
\end{equation}
where $\dim\{\calC_i\}$ is the dimension of $\calC_i$. We note that $\calC_i$ is the column space of $\matP_{(i)}\in\bbR^{N\times N}$ which leads to the following:
\begin{equation}
\dim\{\calC_i\} =\rank{\matP_{(i)}} \leq N-\rank{\matA}\leq s,
\end{equation}
which follows from \eqref{eq:jordan2} of \Cref{lem:jordan} and Condition \ref{con:sparse} of \Cref{thm:necc_suff}. Then, we conclude that \eqref{eq:con_b} holds, and thus, Statement \ref{con:null} is proved. 

Hence, the proof of \Cref{thm:necc_suff} is completed.
\hfill$\blacksquare$
\section{Proof of \Cref{lem:jordan}}\label{app:jordan}
We present a constructive proof which based on the Jordan canonical form~\cite{horn2012matrix} of $\matA$. The number of Jordan blocks in $\matA$ corresponding to the eigenvalue 0 is  $N-\rank{\matA}$, and we assume that the number of such Jordan blocks of size $i$ is $q_i$, i.e., 
\begin{equation}\label{eq:qi_con}
N-\rank{\matA}=\sum_{i=1}^nq_i,
\end{equation}
where $n$ is the size of the largest Jordan block corresponding to the eigenvalue 0. Therefore, the Jordan canonical form of $\matA$ is given by
\begin{equation}
\matA=\matP^{-1}\begin{bmatrix}
\matJ\\
&\matN_{(1)}\\
&&\matN_{(2)}\\
&&&\ddots\\
&&&&\matN_{(n)}
\end{bmatrix}\matP\label{eq:Jordan},
\end{equation}
where $\matP\in\bbR^N$ is a nonsingular matrix, and $\matJ\in\bbR^{q\times q}$ is the nonsingular matrix consisting of the Jordan blocks corresponding to the nonzero eigenvalues of $\matA$ with $q=N-\sum_{i=1}^n iq_i$. Also, $\matN_{(i)}\in\bbR^{iq_i\times iq_i}$ is the nilpotent block diagonal matrix consisting all the $q_i$ Jordan blocks of size $i$ corresponding to the eigenvalues 0, i.e., for any integer $i>0$,
\begin{equation}\label{eq:nilpotent}
\matN_{(i)}=\begin{bmatrix}
\underset{\text{Block } 1}{\underbrace{\begin{matrix}
\zero & \eye_{i-1\times i-1}\\
0 &\zero
\end{matrix}}}\\
&\underset{\text{Block } 2}{\underbrace{\begin{matrix}
\zero & \eye_{i-1\times i-1}\\
0 &\zero
\end{matrix}}}\\
&\hspace{0.5cm}\ddots\\
&&\underset{\text{Block } q_i}{\underbrace{\begin{matrix}
\zero & \eye_{i-1\times i-1}\\
0 &\zero
\end{matrix}}}
\end{bmatrix},
\end{equation}
where $ \eye_{i\times i}\in\bbR^{i\times i}$ denotes the  identity matrix. Therefore, for $k\geq 0$, \eqref{eq:Jordan} gives
\begin{equation}\label{eq:basic_1}
\matP\matA^k=\begin{bmatrix}
\matJ^{k}\\
&\matN_{(1)}^{k}\\
&&\matN_{(2)}^{k}\\
&&\hspace{0.5cm}\ddots\\
&&&\matN_{(n)}^{k}
\end{bmatrix}\matP.
\end{equation}
However, from \eqref{eq:nilpotent},  for $0\leq k<i$, we obtain 
\begin{equation}\label{eq:nilpotent_pow}
\matN_{(i)}^k= \begin{bmatrix}
\zero & \eye_{i-k\times i-k}\\
\zero &\zero\\
&&\zero &   \eye_{i-k\times i-k}\\
&&\zero &\zero\\
&&&\hspace{0.5cm}\ddots\\
&&&&\zero &   \eye_{i-k\times i-k}\\
&&&&\zero &\zero\\
\end{bmatrix}
\end{equation}
Here, the last $k$ rows of each block along the diagonal of $\matN_{(i)}^k$ are zeros, and so, $\matN_{(i)}^k\in\bbR^{iq_i\times iq_i}$ has exactly $kq_i$ all zero rows, for $0\leq k< i$. Thus, $\matP\matA^k$ has $r_k$ more zero rows than $\matP\matA^{k-1}$ where
\begin{equation}\label{eq:ri_defn}
r_k\triangleq \sum_{j=k}^nq_j\leq N-\rank{\matA} ,
\end{equation}
for $1 \leq k\leq n$, and the inequality follows from \eqref{eq:qi_con}. We rearrange the rows of the matrices on the both sides of \eqref{eq:basic_1} so that these $r_k$ zero rows corresponding to each value of $1\leq k\leq n$ are grouped together as follows:
\begin{align}\label{eq:basic_2}
\begin{bmatrix}
\matP_{(0)}\in\bbR^{q\times N}\\
\tilde{\matP}_{(1)}\in\bbR^{r_1\times N}\\
\tilde{\matP}_{(2)}\in\bbR^{r_2\times N}\\
\vdots\\
\tilde{\matP}_{(n)}\in\bbR^{r_n\times N}
\end{bmatrix}\matA^k=\begin{bmatrix}
\begin{matrix}
\matJ^{k}\in\bbR^{q\times q} &\zero \in\bbR^{q\times N-q}
\end{matrix}\\
\zero \in\bbR^{r_1\times N}\\
\vdots\\
\zero \in\bbR^{r_k\times N}\\
\matQ_{(k,k+1)}\in\bbR^{r_{k+1}\times N}\\
\vdots\\
\matQ_{(k,n)}\in\bbR^{r_n\times N}\\
\end{bmatrix}\matP,
\end{align}
where we define the matrices as follows:
\begin{itemize}
\item The matrix $\matP_{(0)}$ is the full row rank submatrix formed by the first $q$ rows of the nonsingular matrix $\matP$.
\item  For $1\leq i\leq n$, the matrix $\tilde{\matP}_{(i)}$ is the full row rank submatrix of the nonsingular matrix $\matP$ formed by the $(i-1)\nth$ row above the last row of each of the $r_i$ Jordan blocks of size $\geq i$ corresponding to the eigenvalue 0.
\item For $k+1\leq j\leq n$, the matrix $\matQ_{(k,j)}$ is the full row rank matrix corresponding to the remaining nonzero rows in $\matN_{(i)}^k$. 
\end{itemize}

Next, we decompose $\matP$ as follows:
\begin{equation}
\matP = \begin{bmatrix}
\matP_{(0)}\in\bbR^{q\times N}\\
\zero\in\bbR^{N-q\times N}
\end{bmatrix} + \begin{bmatrix}
\zero\in\bbR^{\lb q+ \sum_{j=1}^{i-1}r_j\rb\times N}\\
\tilde{\matP}_{(i)}\in\bbR^{r_i\times N}\\
\zero\in\bbR^{\lb \sum_{j=i+1}^{n}r_j\rb\times N}
\end{bmatrix}.
\end{equation}
We define $\matP_{(i)} \triangleq  \begin{bmatrix}
\zero\\
\tilde{\matP}_{(i)}\in\bbR^{r_i\times N}\\
\zero
\end{bmatrix}\in\bbR^{N\times N}$ so that \eqref{eq:jordan1} is satisfied. Therefore, the proof is complete if we show that the matrices $\matP_{(0)}$ and $\matP_{(i)}$ satisfy Conditions \ref{con:jordan1} and \ref{con:jordan2} of \Cref{lem:jordan}, respectively.

To verify  Conditions \ref{con:jordan1}, we note that $\matP_{(0)}\in\bbR^{q\times N}$ is full row rank, and as a result, $\rank{\matP_{(0)}}=q$. Also, from \eqref{eq:basic_2}, we have
\begin{equation}
\begin{bmatrix}
\matP_{(0)}\\
\zero
\end{bmatrix}\matA^k = \begin{bmatrix}
\matJ^k\matP_{(0)}\\
\zero
\end{bmatrix}.
\end{equation}
Therefore, Condition \ref{con:jordan1} holds.

Finally, to verify Condition \ref{con:jordan2}, from \eqref{eq:basic_2}, we get
\begin{align}
\matP_{(i)}\matA^k &= \begin{bmatrix}
\zero\\
\tilde{\matP}_{(i)}\\
\zero
\end{bmatrix}\matA^k=\zero, k\geq i\label{eq:inter3}
\end{align} 
Also, from the definition of $\matP_{(i)}$ and \eqref{eq:basic_2},
\begin{align}
\rank{\matP_{(i)}\matA^{i-1}} \notag\\
&\hspace{-2.7cm}= \rank{
\tilde{\matP}_{(i)}\matA^{i-1}}= \rank{
\tilde{\matP}_{(i)}\matA^{i-1}\matP^{-1}}\\
& \hspace{-2.7cm}= \rank{\matQ_{(i-1,i)}} = r_i = \rank{\tilde{\matP}_{(i)}}\label{eq:inter_1}=\rank{\matP_{(i)}} ,
\end{align}
where \eqref{eq:inter_1} follows because $\matQ_{(i-1,i)},\tilde{\matP}_{(i)}\in\bbR^{r_i\times N}$ are full row rank. Further, from \eqref{eq:ri_defn},
\begin{equation}
\rank{\matP_{(i)}}  = \rank{\matP_{(i)}\matA^{i-1}} = r_i\leq N-\rank{\matA}.\label{eq:inter_2}
\end{equation}
Thus, from \eqref{eq:inter3} and \eqref{eq:inter_2}, we deduce that Condition \ref{con:jordan2} holds. 

Hence, we conclude that all the conditions of \Cref{lem:jordan} hold for our construction and the proof is complete.
\hfill$\blacksquare$
\section{Proof of \Cref{lem:nonzero}}\label{app:nonzero}
Using the Kalman rank type condition in \eqref{eq:combinatorial}, when the system is not sparse controllable
\begin{equation}\label{eq:Ncombinatorial}
\posp\lc \begin{bmatrix}\matA^{K-1} \matB_{\calS_1} &\matA^{K-2} \matB_{\calS_2} \ldots  \matB_{\calS_K}\end{bmatrix}\rc\subset\bbR^N,
\end{equation}
for any positive integer $K<\infty$ and index sets $\calS_k\subset\{1,2,\ldots,m\}$ with $\lv\calS_k\rv\leq s$, for all the integers $1\leq k\leq K$.

Let the sets $\lc\calS_i\subseteq\{1,2,\ldots,m\}\rc_{ i=1}^{ \tilde{K}}$, each with cardinality $s$, be such that they partition the set $\{1,2,\ldots,m\}$ as follows:%\footnote{For example, $\tilde{K} = \lceil m/s\rceil$ and $\calS_i=\{(i-1)s+1,(i-1)s+2,\ldots,\min\{is,m\}\}$. }
\begin{equation}\label{eq:concat}
\lv\calS_i\rv = s \text{ and } \bigcup_{i=1}^{\tilde{K}}\calS_i=\{1,2,\ldots,m\}.
\end{equation}
Then, \eqref{eq:Ncombinatorial} immediately yields that for any positive integer $\tilde{N}<\infty$,
\begin{multline}
\posp\bigg\{
\big[ \begin{matrix}
\matA^{\tilde{K} \tilde{N}-1}\matB_{\calS_1} &  \matA^{\tilde{K} \tilde{N}-2}\matB_{\calS_1} &\ldots 
& \matA^{(\tilde{K}-1) \tilde{N}}\matB_{\calS_1}  \end{matrix}\\
\begin{matrix}
\ldots & \matA^{(\tilde{K}-1) \tilde{N}-1}\matB_{\calS_2}&\ldots \matA^{(\tilde{K}-2) \tilde{N}}\matB_{\calS_2} &\ldots
\end{matrix}\\
\begin{matrix}
\ldots &
  \matA^{\tilde{N}-1}\matB_{\calS_{\tilde{K}}}&\ldots&\matB_{\calS_{\tilde{K}}}
\end{matrix}\big]\bigg\}\subset \bbR^N.
\end{multline}
On rearranging the columns, this is equivalent to the following:
\begin{equation}\label{eq:B*con}
\posp\bigg\{
\begin{bmatrix}
\matA^{\tilde{N}-1}\matB^* & \matA^{\tilde{N}-2}\matB^*&\ldots&\matB^*
\end{bmatrix}\bigg\}\subset \bbR^N,
\end{equation}
where we define $\matB^*\in\bbR^{N\times \tilde{K}s}$ as follows:
\begin{equation}\label{eq:B*defn}
\matB^*\triangleq \begin{bmatrix}
\matA^{(\tilde{K} -1)\tilde{N}}\matB_{\calS_1} & \matA^{(\tilde{K} -2)\tilde{N}}\matB_{\calS_2}  &\ldots&\matB_{\calS_{\tilde{K}}} 
\end{bmatrix}.
\end{equation}
Therefore, the linear dynamical system $(\matA,\matB^*)$ is not controllable using nonnegative controls. Applying \Cref{thm:nonnegative} to the system $(\matA,\matB^*)$, we see that at least one of the two conditions of \Cref{thm:nonnegative} is violated:
\begin{itemize}
\item $\exists(\lambda,\vecz\neq\zero)$ such that 
$\vecz\tran\matA=\lambda\vecz\tran$ and $\vecz\tran\matB^*=\zero$.
\item $\exists(\lambda\geq 0,\vecz\neq\zero)$ such that 
$\vecz\tran\matA=\lambda\vecz\tran$  and $\vecz\tran\matB^*\leq \zero$.
\end{itemize}
Further, if $\vecz\tran\matA=\lambda\vecz\tran$, from \eqref{eq:B*defn}, we obtain
\begin{equation}\label{eq:zerolambda}
\vecz\tran\matB^* = \begin{bmatrix}\lambda^{(\tilde{K}-1)N}\vecz\tran\matB_{\calS_1} & \ldots&\vecz\tran\matB_{\calS_{\tilde{K}}}
\end{bmatrix}.
\end{equation}
Here, $\lambda\neq0$ because $\matA$ is nonsingular. Thus, if $\vecz\tran\matB^*=\zero$, then \eqref{eq:concat} and \eqref{eq:zerolambda} imply that $\vecz\tran\matB=\zero$. Similarly, if $\vecz\tran\matB^*\leq \zero$, then $\vecz\tran\matB\leq\zero$. This leads to at least one of the following:
\begin{itemize}
\item $\exists(\lambda,\vecz\neq\zero)$ such that 
$\vecz\tran\matA=\lambda\vecz\tran$ and $\vecz\tran\matB=\zero$.
\item $\exists(\lambda\geq 0,\vecz\neq\zero)$ such that 
$\vecz\tran\matA=\lambda\vecz\tran$  and $\vecz\tran\matB\leq \zero$.
\end{itemize}
So, if the system $(\matA,\matB)$ is not nonnegative $s-$sparse controllable, at least one of the conditions of \Cref{thm:nonnegative} is violated. This implies that the conditions of \Cref{thm:nonnegative} are sufficient for nonnegative $s-$sparse controllability. 

Hence, the proof is complete.
\hfill$\blacksquare$
\section{Proof of \Cref{lem:union}}\label{app:union}
We first note that $\calZ=\posp\{\matZ\}$ implies that all the columns of $\matZ$ belong to $\calZ$ and they linearly span $\calZ$. As a result, $\rank{\matZ} = d$. Further, for any $\vecz\in\calZ$, we define 
\begin{equation}
\calV_{\vecz}\triangleq \lc\vecalpha\in\bbR^m:\matZ\vecalpha=\vecz \text{ and }\vecalpha\geq \zero\rc.
\end{equation}
 The set $\calV_{\vecz}$ is non-empty due to the relation $\calZ=\posp\{\matZ\}$. Let $\bar{\matZ}\vecalpha=\bar{\vecz}$ be the reduced-row echelon form of the matrix equation $\matZ\vecalpha=\vecz$, after removing the zero rows. Consequently, we obtain
\begin{equation}
\calV_{\vecz}= \lc\vecalpha\in\bbR^m:\bar{\matZ}\vecalpha=\bar{\vecz} \text{ and }\vecalpha\geq 0\rc.
\end{equation}
By the fundamental theorem in linear programming, the system $\bar{\matZ}\vecalpha=\bar{\vecz}$ has a basic feasible solution $\vecalpha\geq \zero$ with at most $\rank{\bar{\matZ}}$ nonzero elements. However, we note that
%\begin{equation}
$\rank{\bar{\matZ}}=\rank{\matZ}= d$.
%\end{equation}
 As a consequence, there exists a $d-$sparse vector $\vecalpha\geq\zero$ such that $\vecz=\matZ\vecalpha.$ So, we conclude that \eqref{eq:union} holds.
%\begin{equation}
%\vecz\in \bigcup_{\calS\subseteq\{1,2,\ldots,m\},\lv\calS\rv=d}\posp\{\matZ_\calS\}.
%\end{equation}
%Since the above result holds for any $\vecz\in\calZ$, we deduce that
%\begin{equation}\label{eq:upper}
%\calZ\subseteq \bigcup_{\calS\subseteq\{1,2,\ldots,m\},\lv\calS\rv=d}\posp\{\matZ_\calS\}
%\subseteq\posp\{\matZ\}=\calZ.
%\end{equation}
Hence, the proof is complete.
\hfill$\blacksquare$
\section{Proof of \Cref{cor:sbound}}\label{app:sbound}
Since the system is nonnegative controllable, from \Cref{thm:nonnegative}, Conditions \ref{con:control} and \ref{con:nonnegative} of \Cref{thm:necc_suff} hold. Accordingly, it suffices to check if Condition \ref{con:sparse} holds. 

Let $\calN$ be the null space of $\matA\tran$. The nonnegative controllability of the system implies that for any $\vecx\in\calN$, there exists a positive integer $K<\infty$ and $\lc\vecu_k\in\bbR^N_+\rc_{k=1}^K$ such that 
\begin{equation}
\vecx=\sum_{k=1}^{K}\matA^{K-k}\matB\vecu_k.
\end{equation}
Multiplying both the sides of above equation with the orthogonal projection operator $\matA^\dagger\in\bbR^{N\times N}$ of $\calN$, we get
\begin{equation}\label{eq:null_con1}
\vecx = \matA^{\dagger}\vecx = \matA^{\dagger}\matB\vecu_K,
\end{equation}
because $\matA^\dagger\matA=\zero$. Now, from \eqref{eq:null_con1},  we deduce that 
\begin{equation}
\calN\subseteq \posp\lc\matA^{\dagger}\matB\rc.
\end{equation}

Further, the minimum number of vectors required to positively span a subspace is one more than its dimension~\cite[Corollary 5.5]{regis2016properties}. Since the dimension of $\calN$ is $N-\rank{\matA}$, we have $m\geq N-\rank{\matA}+1$. Therefore,
\begin{equation}
s\geq m-1 \geq  N-\rank{\matA}.
\end{equation} 
Hence, Condition \ref{con:sparse} of \Cref{thm:necc_suff} also holds, and the proof is complete.\hfill$\blacksquare$

\bibliographystyle{IEEEtran}
\bibliography{Supporting_Files/IEEEabrv,Supporting_Files/bibJournalList,Supporting_Files/PControl_cite}
\end{document}